\documentclass[oneside,english,reqno,12pt]{amsart}
\usepackage[12pt]{extsizes}

\calclayout
\usepackage{geometry}
\usepackage{babel}
\usepackage{enumitem}
\usepackage{amssymb,amsmath,amsfonts,latexsym}
\usepackage{bbm}
\usepackage{appendix}
\DeclareMathOperator{\tr}{Tr}
\usepackage{hyperref}
\usepackage{color}
\makeatletter
\numberwithin{equation}{section}
\numberwithin{figure}{section}
\theoremstyle{plain}
\newtheorem{theorem}{\protect\theoremname}
\theoremstyle{corollary}

\theoremstyle{remark}
\newtheorem{remark}[theorem]{\protect\remarkname}
\theoremstyle{plain}
\newtheorem{lemma}[theorem]{\protect\lemmaname}
\newlist{casenv}{enumerate}{4}
\setlist[casenv]{leftmargin=*,align=left,widest={iiii}}
\setlist[casenv,1]{label={{\itshape\ \casename} \arabic*.},ref=\arabic*}
\setlist[casenv,2]{label={{\itshape\ \casename} \roman*.},ref=\roman*}
\setlist[casenv,3]{label={{\itshape\ \casename\ \alph*.}},ref=\alph*}
\setlist[casenv,4]{label={{\itshape\ \casename} \arabic*.},ref=\arabic*}

\makeatother

\providecommand{\corollaryname}{Corollary}
\providecommand{\lemmaname}{Lemma}
\providecommand{\remarkname}{Remark}
\providecommand{\casename}{Case}
\providecommand{\theoremname}{Theorem}

\begin{document}

\title[Many-Body Blow-Up of Boson Stars]{Many-Body Blow-Up Profile of Boson Stars \\ with External Potentials}


\begin{abstract}
We consider a 3D quantum system of $N$ identical bosons in a trapping potential $|x|^p$, with $p\geq0$, interacting via a Newton potential with an attractive interaction strength $a_{N}$. For a fixed large $N$ and the coupling constant $a_{N}$ smaller than a critical value $a_{*}$ \emph{(Chandrasekhar limit mass)}, in an appropriate sense, the many-body system admits a ground state. We investigate the blow-up behavior of the ground state energy as well as the ground states when $a_{N}$ approaches $a_{*}$ sufficiently slowly in the limit $N\to\infty$. The blow-up profile is given by the Gagliardo--Nirenberg solutions.
\end{abstract}

\author{Dinh-Thi Nguyen}
\address{Dinh-Thi Nguyen, Mathematisches Institut, Ludwig--Maximilans--Universit\"at M\"unchen, Theresienstrasse 39, 80333 Munich, Germany.} 
\email{nguyen@math.lmu.de}

\maketitle

\section{Introduction}
The so-called \emph{boson stars} are a class of models from relativistic many-body quantum mechanics inspired by the Chandrasekhar theory of stellar collapse. For fermion stars like white dwarfs and neutron stars, the collapse under attractive gravitational forces is prevented by Pauli exclusion principles. For boson stars, such a toy model of gravitational collapse is obtained by neglecting the Pauli principle (as also considered in the so-called \emph{bosonic atoms} \cite{BeLi-83,So-90}), i.e. assuming that all the force carrier particles are bosons. It is a fundamental fact that boson stars collapse when their masses are bigger than a critical number. The maximum mass of stable stars, called the \emph{Chandrasekhar limit mass}, was discovered by Chandrasekhar \cite{Ch-31-1,Ch-31-2,Ch-31-3,Ch-84} in 1930, which gained him the 1983 Nobel Prize in Physics. In this paper, we will study the collapse phenomenon of boson stars with a rigorous mathematical approach. 

We consider a system of $N$ identical relativistic bosons in $\mathbb{R}^3$ interacting via the self-gravitational force, described by the \emph{pseudo-relativistic} Schr\"{o}dinger Hamiltonian
\begin{equation}\label{hamiltonian}
H_{N} = \sum_{i=1}^{N}\big(\sqrt{-\Delta_{x_i}+m^2}+V(x_i)\big)-\frac{a_{N}}{N-1}\sum_{1\leq i<j \leq N}|x_i-x_j|^{-1}
\end{equation}
acting on $\mathfrak{H}^N=\bigotimes_{\rm sym}^N L^2(\mathbb{R}^3)$, the Hilbert space of square-integrable symmetric functions. The pseudo-differential operator $\sqrt{-\Delta+m^{2}}$ describes the kinetic energy of a relativistic quantum particle with mass $m > 0$. The case $V=0$ is the most physically relevant, but it is also mathematically interesting to include a general external potential $V$. The parameter $a_{N}>0$ describes the strength of the interaction. We will take $a_{N}\uparrow a_{*}$ for a critical value $a_{*}$ described below. The coupling constant $1/(N-1)$ ensures that the kinetic and interaction energies are comparable in the limit $N\to\infty$.

We are interested in the large-$N$ behavior of the ground states of the system. Recall that the ground state energy per particle of $H_{N}$ is defined by
$$
E_{N}^{\rm Q} := N^{-1}\inf\text{spec } H_{N} = N^{-1}\inf_{\Psi_{N}\in\mathfrak{H}^N,\|\Psi_{N}\|_{L^2}=1} \left\langle \Psi_{N},H_{N}\Psi_{N} \right\rangle.
$$
Note that without the external potential, the system is translation invariant and there are no ground states. Moreover, in this case, the condensation of the {\em approximate ground states} (in the energy sense) is not expected, because the system can easily split into many small pieces with same energy. An external potential $V$ enforces the existence of ground states of $H_{N}$. 

We will assume that $a_{N}<a_{*}$, where $a_{*}$ is the optimal constant of the Gagliardo--Nirenberg-type inequality
\begin{equation}\label{ineq:GN}
\|(-\Delta)^{\frac{1}{4}}u\|_{L^{2}}^{2} \|u\|_{L^{2}}^{2} \geq \frac{a_{*}}{2}\iint_{\mathbb{R}^{3}\times\mathbb{R}^{3}}\frac{|u(x)|^{2}|u(y)|^{2}}{|x-y|}{\rm d}x{\rm d}y,\quad \forall u\in H^{\frac{1}{2}}(\mathbb{R}^{3}).
\end{equation}
It is well-known (see \cite{LiYa-87,LiTh-84,Le-07,FrJoLe-07,FrLe-09,LeLe-11}) that $4/\pi<a_{*}<2.7$ and that the inequality \eqref{ineq:GN} has an optimizer $Q\in H^{\frac{1}{2}}(\mathbb{R}^3)$. This optimizer is positive radially symmetric decreasing, by rearrangement inequality (see \cite[Chapter 3]{LiLo}), and can be chosen to satisfy
\begin{equation}\label{eq:GN-optimizer}
\|(-\Delta)^{\frac{1}{4}}Q\|_{L^{2}}^{2} =  \|Q\|_{L^{2}}^{2} = \frac{a_{*}}{2}\iint_{\mathbb{R}^{3}\times\mathbb{R}^{3}}\frac{|Q(x)|^{2}|Q(y)|^{2}}{|x-y|}{\rm d}x{\rm d}y=1.
\end{equation}
Moreover, such $Q$ solves the \emph{massless boson star equation}
\begin{equation}
\sqrt{-\Delta}Q+Q-a_{*}(\left|\cdot\right|^{-1}\star|Q|^{2})Q=0\label{eq:massless boson star}
\end{equation}
and it verifies the decay properties (see \cite{FrLe-09})
\begin{equation} \label{eq:decay}
Q(x)\le C (1+|x|)^{-4} \quad \text{and} \quad (\left|\cdot\right|^{-1}\star|Q|^{2})(x)\leq C(1+|x|)^{-1}.
\end{equation}
The uniqueness (up to translations and dilations) of the optimizers of \eqref{ineq:GN}, as well as the uniqueness (up to translations) of the positive solutions of \eqref{eq:massless boson star}, are {\em open problems} (see \cite{Le-09,FrLe-09,GuZe-19} for related discussions). Note that the translations and dilations can be fixed by using \eqref{eq:GN-optimizer} and the radial properties. In the following, we define 
\begin{equation}\label{cG}
\mathcal{G} = \left\{\text{positive radially decreasing functions satisfying \eqref{eq:GN-optimizer}--\eqref{eq:massless boson star}}\right\} . 
\end{equation}
It is expected that $\mathcal{G}$ has only one element (see \cite{LiYa-87}), but  our analysis will not rely on this conjecture. 

In a seminal paper \cite{LiYa-87}, Lieb and Yau proved that for any fixed $a_{N}=a<a_{*}$, the quantum energy converges to the semi-classical Hartree energy 
\begin{equation}\label{fundamental}
\lim_{N\to\infty}E_{N}^{\rm Q} =  \inf_{u\in H^{\frac{1}{2}}(\mathbb R^3),\|u\|_{L^2}=1}\mathcal{E}_{a}^{\rm H}(u)=: E_{a}^{\rm H} 
\end{equation}
where
\begin{equation}\label{eq:boson star functional}
\mathcal{E}_{a}^{\rm H} (u) = \|(-\Delta+m^{2})^{\frac{1}{4}}u\|_{L^{2}}^2 + \int_{\mathbb{R}^{3}}V(x)|u(x)|^{2}{\rm d}x -\frac{a}{2}\iint_{\mathbb{R}^{3}\times\mathbb{R}^{3}}\frac{|u(x)|^{2}|u(y)|^{2}}{|x-y|}{\rm d}x{\rm d}y.
\end{equation}
In fact, the Hartree energy functional $\mathcal{E}_{a}^{\rm H}$ is obtained by assuming that all the particles occupy a common one-particle state, and hence the Hartree energy $E_{a}^{\rm H}$ is an upper bound to the quantum energy $E_{N}^{\rm Q}$.

We note that Lieb and Yau proved \eqref{fundamental} without external potentials, but their result can be extended easily with an external potential $V$. A new proof of \eqref{fundamental} was found recently by Lewin, Nam and Rougerie \cite{LeNaRo-14}, using the \emph{quantum de Finetti} theorems \cite{St-69,HuMo-76,FaVa-06,ChKoMiRe-07,Ch-11,Ha-13,LeNaRo-14,LeNaRo-15,LeNaRo-15-fr,Ro-15,LeNaRo-16,NaRoSe-16,LeNaRo-17,BrHa-17,LeNaRo-18}. This approach allows us to prove, in the case of trapping potentials, the condensation of the many-body ground states to the Hartree ground states in terms of reduced density matrices
$$
\lim_{N\to\infty} \tr \Big|\gamma_{\Psi_N}^{(k)}-\int |u^{\otimes k}\rangle\langle u^{\otimes k} | {\rm d} \mu(u) \Big| = 0, \quad \forall k\in \mathbb{N}.
$$
Here ${\rm d}\mu$ is a Borel probability measure supported on the Hartree ground states. Recall that the $k$-particle reduced density matrix, defined for any $\Psi\in\mathfrak{H}^N$, is the  partial trace
$$
\gamma_{\Psi}^{(k)}:=\tr_{k+1\to N}|\Psi\rangle\langle\Psi|.
$$
Equivalently, $\gamma_{\Psi}^{(k)}$ is the trace class operator on $\mathfrak{H}^k$ with kernel
$$
\gamma_{\Psi}^{(k)}(x_1,\ldots,x_k;y_1,\ldots,y_k)=\int_{\mathbb{R}^{3(N-k)}}\overline{\Psi(x_1,\ldots,x_k;Z)}{\Psi(y_1,\ldots,y_k;Z)}{\rm d}Z.
$$

In the present paper, we will study the blow-up phenomenon of boson stars when $a_{N} \uparrow a_{*}$ as $N\to\infty$. Due to the technical reason explained above, we will assume the presence of an external potential, which is eventually taken for simplicity under the form  
$$
V(x)=|x|^p
$$
for a fixed parameter $p > 0$. This ensures the existence of ground states of $H_{N}$ when $a_{N}<a_{*}$, by a standard compactness argument. We will prove that the ground states have a universal blow-up profile described by the set of Gagliardo--Nirenberg optimizers. 

To state our result, let us introduce the following notations
\begin{equation}\label{eq:scaling-length}
\ell_{N}=\Lambda (a_{*}-a_{N})^{-\frac{1}{q+1}}
\end{equation}
where $q=\min\{p,1\}$ and 
\begin{equation} \label{lambda-p}
\Lambda = \begin{cases} \displaystyle \inf_{W\in \mathcal{G}}\Big(a_{*} p \int_{\mathbb{R}^{3}}|x|^{p}|W(x)|^{2}{\rm d}x\Big)^{\frac{1}{p+1}} & \text{if } 0<p<1, \\ 
\displaystyle  \inf_{W\in \mathcal{G}}\Big(\frac{m^{2}a_{*}}{2}\|(-\Delta)^{-\frac{1}{4}}W\|_{L^{2}}^{2}+a_{*} \int_{\mathbb{R}^{3}}|x||W(x)|^{2}{\rm d}x\Big)^{\frac{1}{2}} & \text{if  }p=1,\\
\displaystyle  \inf_{W\in \mathcal{G}}m\sqrt{\frac{a_{*}}{2}}\|(-\Delta)^{-\frac{1}{4}}W\|_{L^{2}} & \text{if  } p>1.
\end{cases}
\end{equation}
In the case $p\geq 1$, the collapse scale \eqref{lambda-p} is set by the subleading contribution of the kinetic energy in a large momentum expansion. More precisely, it comes from the second order in the Taylor expansion of $\sqrt{-\Delta+m^2}$ (see e.g. \cite{Ng-17-boson} for more details). This effect does not appear in \cite{GuSe-14,LeNaRo-17-book} (see also \cite{EyRo-19} for a related result where a subleading contribution from the high momentum expansion of the interaction is relevant).

Our main result is the following

\begin{theorem}[Collapse and condensation of the many-body ground states]\label{thm:behavior-Nbody} 
	Assume that $m>0$ and $V(x)=|x|^p$ for some $p > 0$. Let $a_{N} = a_{*} - N^{-\alpha}$ with $0<\alpha <1/3$. Then we have
	\begin{equation}\label{asymptotic:many-body-energy-1}
	E_{N}^{\rm Q} = (a_{*}-a_{N})^{\frac{q}{q+1}} \Big(\frac{q+1}{q}\cdot\frac{\Lambda}{a_{*}} + o(1)_{N\to\infty}\Big)
	\end{equation}
	where $q=\min\{p,1\}$ and $\Lambda$ is given by \eqref{lambda-p}.
	
	In addition, assume that $0<p\leq 1$ and $0<\alpha<p/(17p+15)$. Let $\Psi_N$ be a ground state of $H_{N}$, which exists. Then there exists a Borel probability measure ${\rm d}\mu$ supported on $\mathcal{G}$ defined in \eqref{cG} such that, along a subsequence of the rescaled states $\Phi_N=\ell_{N}^{-3N/2}\Psi_N(\ell_{N}^{-1}\cdot)$, we have
	\begin{equation}\label{conv:bec}
	\lim_{N\to\infty} \tr \Big|\gamma_{\Phi_N}^{(k)}-\int |u^{\otimes k}\rangle\langle u^{\otimes k}|{\rm d}\mu(u)\Big| = 0, \quad \forall k\in \mathbb{N}.
	\end{equation}
\end{theorem}

\begin{remark}
	\begin{itemize}
		\item If $V \equiv 0$, then we also obtain \eqref{asymptotic:many-body-energy-1} with $q=1$ and $\Lambda$ is given by the case $p>1$ of \eqref{lambda-p}.
		\item If $\mathcal{G}=\{Q_{0}\}$  (as conjectured in \cite{LiYa-87}), then for $p>0$ and $0<\alpha<1/3$ we have
		$$
		\lim_{N\to\infty} \tr \big|\gamma_{\Phi_N}^{(k)}-|Q_{0}^{\otimes k}\rangle\langle Q_{0}^{\otimes k}|\big| = 0, \quad \forall k\in \mathbb{N},
		$$
		without the constraints $p\leq 1$ and $\alpha<p/(17p+15)$. Moreover, the convergence holds for the whole sequence as $N\to\infty$.
	\end{itemize}
\end{remark}

The asymptotic behavior of the quantum energy follows from that of the Hartree energy. The energy estimate essentially follows from the analysis of Lieb and Yau \cite{LiYa-87}. Our main result in Theorem \ref{thm:behavior-Nbody} relies on the convergence of the ground states. Note that we obtain \eqref{asymptotic:many-body-energy-1} for all $p>0$. However, because of the lack of some compactness in the case $p>1$, which is somewhat similar to the translation-invariant case when $V \equiv 0$, our result on the Bose--Einstein condensation \eqref{conv:bec} is restricted to the case $0<p\leq 1$.

Our work is inspired by the recent study by Lewin, Nam and Rougerie \cite{LeNaRo-17-book} on the mass concentration of the Bose--Einstein condensate described by the 2D focusing many-body systems. In this case, the Gagliardo--Nirenberg inequality (with the usual non-relativistic kinetic energy and the local non-linear interaction energy)  has a unique optimizer. In particular, the convergence of the ground states essentially follows from the convergence of the ground state energy via a Feynman--Hellman argument, which simplifies the analysis and allows us to improve the range for the parameters $p$ and $\alpha$, but at the cost of requiring the uniqueness of the Gagliardo--Nirenberg optimizer.
However, in our model, the kinetic operator is non-local and the uniqueness of the limit profile of the ground states is unknown, which complicates our analysis in many places. Without this uniqueness, the method in \cite{LeNaRo-17-book} can still be coupled with the \emph{quantum de Finetti} theorem and moments estimates to derive the convergence of the ground state energy. This approach will be used in our paper. More precisely, we will reduce the problem to a finite dimensional setting, and then employ the \emph{quantitative} version of quantum de Finetti \cite{ChKoMiRe-07,Ch-11,Ha-13,FaVa-06,LeNaRo-15,LeNaRo-16,LeNaRo-17} with a refined relativistic estimate. The second moment estimate is a classical idea, which goes back to Erd\H{o}s and Yau \cite{ErYa-01} (see also \cite{ErScYa-10,NaRoSe-16,LeNaRo-17}).

So far, in the literature, the blow-up analysis of boson stars has been carried out only in the level of Hartree theory \cite{GuZe-17,Ng-17-H,Ng-17-boson,YaYa-17}, see also Section \ref{sec:blowup approximate} for a review. Here we give the first analysis from the full many-body level, which are significantly more difficult. Finally, we remark that the blow-up profile of neutron stars has been also studied in the Hartree--Fock--Bogoliubov and Chandrasekhar theories \cite{LeLe-10,Ng-17-neutron,Ng-19-neutron}. See also \cite{FrJoLe-07,FrJoLe-07-dy,FrLe-07-boson,Le-07,ElSc-07,LeLe-11,MiSc-12,FrLe-07-neutron,HaSc-09,Ha-10,HaLeLeSc-10} for rigorous results of the dynamical collapse of boson and fermion stars in the time-dependent setting.

\medskip
\textbf{Organization of the paper.} In Section \ref{sec:blowup approximate} we revisit the blow-up phenomenon in Hartree theory. In Section \ref{sec:energy estimates}, we establish energy estimates and moments estimates for the ground state energy and the many-body ground states. The proof of Theorem \ref{thm:behavior-Nbody} is concluded in Section \ref{sec:many body blowup}. Appendix \ref{app:operator-bound} contains a density argument and a proof of operator bounds for two-particle interactions.

\section{\label{sec:blowup approximate}Blow-up in the Hartree theory}

In this section, we revisit the blow-up phenomenon for the Hartree problem
$$
E_{a}^{\rm H}=\inf_{u\in H^{\frac{1}{2}}(\mathbb{R}^3),\|u\|_{L^2}=1}\mathcal{E}_{a}^{\rm H} (u).
$$
where $\mathcal{E}_{a}^{\rm H}$ is given in \eqref{eq:boson star functional}. For the reader's convenience, we recall the following results from \cite{LiYa-87,GuZe-17,Ng-17-H,Ng-17-boson,YaYa-17}.

\begin{theorem}[Existence of the Hartree ground states] \label{thm:existence of ground states} 
	Assume that $m>0$ and $V$ satisfies $0\leq V\in L_{{\rm loc}}^{\infty}(\mathbb{R}^{3})$ and $\lim_{|x|\to\infty}V(x)=\infty$. Then the following statements hold true
	\begin{itemize}
		\item[(i)] If $a > a_{*}$ then $E_{a}^{\rm H} = -\infty$. 
		\item[(ii)] If $a = a_{*}$ then $E_{a}^{\rm H} = \inf_{x\in\mathbb{R}^{3}}V(x)$, but it has no ground states. 
		\item[(iii)] If $0<a<a_{*}$ then $E_{a}^{\rm H}>0$ and it has at least one ground state.
		Moreover 
		$$
		\lim_{a\uparrow a_{*}}E_{a}^{\rm H}=E_{a_{*}}^{\rm H}=\inf_{x\in \mathbb R^3} V(x).
		$$
	\end{itemize}
\end{theorem}

\begin{theorem}[Blow-up of the Hartree ground states]\label{thm:behavior-hartree-ground-states} 
	Assume that $m>0$ and $V(x)=|x|^p$ for some $p > 0$. Then for every sequence $\{a_{k}\}$ with $a_{k} \uparrow a_{*}$ as $k\to\infty$, we have 
	\begin{equation}\label{asymptotic:hartree-energy}
	E_{a_{k}}^{\rm H} = (a_{*}-a_{k})^{\frac{q}{q+1}} \Big(\frac{q+1}{q}\cdot\frac{\Lambda}{a_{*}} + o(1)_{k\to\infty}\Big)
	\end{equation}
	where $q=\min\{p,1\}$ and $\Lambda$ is given by \eqref{lambda-p}.
	
	In addition, assume that $u_{k}$ is a non-negative ground state of $E_{a_{k}}^{\rm H}$ for each $0<a_{k}<a_{*}$. Then there exist a subsequence of $\{a_{k}\}$ (still denoted by $\{a_{k}\}$) and a $Q\in \mathcal{G}$ such that the following strong convergences hold true in $H^{\frac{1}{2}}(\mathbb{R}^{3})$.
	\begin{itemize}
		\item[(i)] If $0<p\leq 1$ then
		\begin{equation} \label{blowup-p-1}
		\lim_{k\to\infty}\Lambda^{-\frac{3}{2}}(a_{*}-a_{k})^{\frac{3}{2(p+1)}}u_{k}(\Lambda^{-1}(a_{*}-a_{k})^{\frac{1}{p+1}}x) = Q(x).
		\end{equation}
		\item[(ii)] If $p>1$ then there exists a sequence $\{y_k\}\subset\mathbb R^3$ such that
		\begin{equation} \label{blowup-p-2}
		\lim_{k\to\infty}\Lambda^{-\frac{3}{2}}(a_{*}-a_{k})^{\frac{3}{4}}u_{k}(y_k+\Lambda^{-1}(a_{*}-a_{k})^{\frac{1}{2}}x) = Q(x).
		\end{equation}
	\end{itemize}

	Furthermore, if $V \equiv 0$ then we also obtain the convergence of the Hartree energy as in \eqref{asymptotic:hartree-energy} with $q=1$ and $\Lambda$ is given by the case $p>1$ of \eqref{lambda-p}, as well as the same convergence of the Hartree ground states as in \eqref{blowup-p-2}. Moreover, the optimal $W$ in \eqref{lambda-p} coincides with the $Q$ in \eqref{blowup-p-1} and \eqref{blowup-p-2}
\end{theorem}

	In \cite{GuZe-17,Ng-17-H,Ng-17-boson,YaYa-17}, the proof of Theorem \ref{thm:behavior-hartree-ground-states} is based on a detailed analysis of the Euler--Lagrange equation for the \emph{exact} Hartree ground states. The aim of this section is to extend this result to the general \emph{approximate} Hartree ground states, which will be needed in the proof of Theorem \ref{thm:behavior-Nbody}. We have the following

\begin{theorem}[Blow-up of the approximate Hartree ground states] \label{thm:behavior-approximate-hartree-ground-states}
	Assume that $m>0$ and $V(x)=|x|^p$ for some $0<p\leq 1$. Let $a_{k} \uparrow a_{*}$ as $k\to\infty$ and let $u_{k} \in H^{\frac{1}{2}}(\mathbb{R}^3)$ be a sequence of non-negative functions such that $\|u_{k}\|_{L^2}=1$ and 
	\begin{equation}\label{eq:hartree-approximate-1}
	\lim_{k\to\infty}\frac{\mathcal{E}_{a_{k}}^{\rm H} (u_{k})}{E_{a_{k}}^{\rm H}}=1.
	\end{equation}
	Then there exists a $Q\in\mathcal{G}$ such that, up to extraction of a subsequence,
	$$
	\lim_{k\to\infty}\Lambda^{-\frac{3}{2}}(a_{*}-a_{k})^{\frac{3}{2(p+1)}}u_{k}(\Lambda^{-1}(a_{*}-a_{k})^{\frac{1}{p+1}}x)=Q(x)
	$$
	strongly in $L^{2}(\mathbb R^3)$, where $\Lambda$ is given by \eqref{lambda-p}.
\end{theorem}

\begin{proof}
	Denote $\tilde{u}_k=\ell_{k}^{-\frac{3}{2}}u_{k}(\ell_{k}^{-1}\cdot)$ with $\ell_{k} = \Lambda (a_{*}-a_{k})^{-\frac{1}{p+1}}$, then $\|\tilde{u}_k\|_{L^{2}}=1$. By the interpolation inequality \eqref{ineq:GN}, we have
	\begin{equation}\label{eq:hartree-approximate-2}
	\mathcal{E}_{a_{k}}^{\rm H}(u_{k}) = \mathcal{E}_{a_{k}}^{\rm H}(\ell_{k}^{\frac{3}{2}}\tilde{u}_k(\ell_{k}\cdot)) \geq \ell_{k}\frac{a_{*}-a_{k}}{a_{*}}\|(-\Delta)^{\frac{1}{4}}\tilde{u}_k\|_{L^{2}}^2+\frac{1}{\ell_{k}^{p}}\int_{\mathbb{R}^{3}}V(x)|\tilde{u}_k(x)|^{2}{\rm d}x.
	\end{equation}
	Combining \eqref{eq:hartree-approximate-2} with the assumption \eqref{eq:hartree-approximate-1} and the asymptotic formula of $E_{a_{k}}^{\rm H}$ in \eqref{asymptotic:hartree-energy}, we deduce that
	\begin{equation}\label{conv:approximate}
	\frac{p+1}{p} \cdot \frac{\Lambda}{a_{*}} + o(1)_{k\to\infty} \geq \frac{\Lambda}{a_{*}}\|(-\Delta)^{\frac{1}{4}}\tilde{u}_k\|_{L^{2}}^2 + \frac{1}{\Lambda^p}\int_{\mathbb{R}^{3}}V(x)|\tilde{u}_k(x)|^{2}{\rm d}x.
	\end{equation}
	This implies that $\left\langle\tilde{u}_k,(\sqrt{-\Delta}+V)\tilde{u}_k\right\rangle$ is bounded uniformly in $k$. Since $\sqrt{-\Delta}+V$ has compact resolvent, we deduce from the Banach--Alaoglu theorem and Sobolev's embedding that, up to extraction of a subsequence, $\tilde{u}_k$ converges to a function $W$ weakly in $H^{\frac{1}{2}}(\mathbb{R}^{3})$, strongly in $L^{r}(\mathbb{R}^{3})$ for $2\leq r<3$ and pointwise almost everywhere in $\mathbb{R}^{3}$. In particular, we have $\|W\|_{L^2}=1$ since $\tilde{u}_k\to W$ strongly in $L^{2}(\mathbb{R}^{3})$. Moreover, we have
	\begin{equation}\label{eq:boson-star-optimizer-0}
	\|(-\Delta)^{\frac{1}{4}}\tilde{u}_k\|_{L^{2}}^2-\frac{a_{k}}{2}\iint_{\mathbb{R}^{3}\times\mathbb{R}^{3}}\frac{|\tilde{u}_k(x)|^2|\tilde{u}_k(y)|^{2}}{|x-y|}{\rm d}x{\rm d}y \leq \ell_{k}^{-1} \mathcal{E}_{a_{k}}^{\rm H}(u_{k}).
	\end{equation}
	By taking the limit $k\to\infty$ in \eqref{eq:boson-star-optimizer-0} and using Fatou's lemma and the Hardy--Littlewood--Sobolev inequality (see \cite[Theorem 4.3]{LiLo}), we arrive at
	\begin{equation}\label{eq:boson-star-optimizer-1}
	\|(-\Delta)^{\frac{1}{4}}W\|_{L^{2}}^2-\frac{a_{*}}{2}\iint_{\mathbb{R}^{3}\times\mathbb{R}^{3}}\frac{|W(x)|^2|W(y)|^{2}}{|x-y|}{\rm d}x{\rm d}y \leq 0.
	\end{equation}
	Here we have used the assumption \eqref{eq:hartree-approximate-1} and the asymptotic formula of $E_{a_{k}}^{\rm H}$ in \eqref{asymptotic:hartree-energy} for the term on the right hand side of \eqref{eq:boson-star-optimizer-0}. In view of \eqref{ineq:GN} and the fact that $\|W\|_{L^2}=1$, the equality in  \eqref{eq:boson-star-optimizer-1} must occur. Thus $W$ is an optimizer for \eqref{ineq:GN} and it satisfies \eqref{eq:massless boson star}, after suitable scaling. Hence
	$$
	W(x)=b^{\frac{3}{2}}\mathcal{Q}(bx+x_0)
	$$
	for some $b>0$, $x_0\in\mathbb R^3$, and for $\mathcal{Q}\in H^{\frac{1}{2}}(\mathbb{R}^{3})$ a positive radially symmetric decreasing solution of the equation \eqref{eq:massless boson star}.
	
	We will show that $\mathcal{Q}\in\mathcal{G}$ and that $W \equiv \mathcal{Q}$. We first deduce from $\|W\|_{L^2}=1$ and the equality in \eqref{eq:boson-star-optimizer-1} that $\|\mathcal{Q}\|_{L^2}=1$ and $\mathcal{Q}$ satisfies
	\begin{equation}\label{eq:boson-star-optimizer-2}
	\|(-\Delta)^{\frac{1}{4}}\mathcal{Q}\|_{L^{2}}^2 = \frac{a_{*}}{2}\iint_{\mathbb{R}^{3}\times\mathbb{R}^{3}}\frac{|\mathcal{Q}(x)|^2|\mathcal{Q}(y)|^{2}}{|x-y|}{\rm d}x{\rm d}y.
	\end{equation}
	Since $\mathcal{Q}$ solves the equation \eqref{eq:massless boson star}, we then deduce from \eqref{eq:massless boson star}, \eqref{eq:boson-star-optimizer-2} and $\|\mathcal{Q}\|_{L^2}=1$ that $\mathcal{Q}$ satisfies \eqref{eq:GN-optimizer}. Hence, $\mathcal{Q}\in\mathcal{G}$. Now we prove $W\equiv \mathcal{Q}$ by showing that $b=1$ and $x_0=0$. Taking the limit $k\to\infty$ in \eqref{conv:approximate} we get
	\begin{align}
	\frac{p+1}{p} \cdot \frac{\Lambda}{a_{*}} & \geq \frac{\Lambda}{a_{*}}\|(-\Delta)^{\frac{1}{4}}W\|_{L^{2}}^2 + \frac{1}{\Lambda^p}\int_{\mathbb{R}^{3}}V(x)|W(x)|^{2}{\rm d}x \nonumber \\
	& = \frac{\Lambda b}{a_{*}}\|(-\Delta)^{\frac{1}{4}}\mathcal{Q}\|_{L^{2}}^2 + \frac{1}{\Lambda^p b^p}\int_{\mathbb{R}^{3}}V(x-x_0)|\mathcal{Q}(x)|^{2}{\rm d}x. \label{eq:boson star optimizer-3}
	\end{align}
	Note that $\|(-\Delta)^{\frac{1}{4}}\mathcal{Q}\|_{L^{2}}=1$, by \eqref{eq:GN-optimizer}. Moreover,
	\begin{equation}\label{ineq:rearrangement}
	\int_{\mathbb{R}^{3}}V(x-x_0)|\mathcal{Q}(x)|^{2}{\rm d}x \geq \int_{\mathbb{R}^{3}}V(x)|\mathcal{Q}(x)|^{2}{\rm d}x
	\end{equation}
	by the rearrangement inequality as $\mathcal{Q}$ is symmetric decreasing and $V$ is strictly symmetric increasing (see \cite[Chapter 3]{LiLo}). Thus
	\begin{equation}\label{eq:boson star optimizer-4}
	\frac{p+1}{p} \cdot \frac{\Lambda}{a_{*}} \geq \frac{\Lambda b}{a_{*}} + \frac{1}{\Lambda^p b^p}\int_{\mathbb{R}^{3}}V(x)|\mathcal{Q}(x)|^{2}{\rm d}x.
	\end{equation}
	On the other hand, it is elementary to check that
	$$
	\inf_{t>0}\Big(\frac{t}{a_{*}} + \frac{1}{t^p}\int_{\mathbb{R}^{3}}V(x)|\mathcal{Q}(x)|^{2}{\rm d}x\Big) = \frac{p+1}{p} \cdot \frac{\Lambda}{a_{*}}
	$$
	with the unique optimal value $t=\Lambda$. Therefore, the equality in \eqref{eq:boson star optimizer-4} must occur. Hence $b=1$. This also implies that the equality in \eqref{ineq:rearrangement} must occur, and $x_0=0$.
\end{proof}

\begin{remark}
	The result in Theorem \ref{thm:behavior-approximate-hartree-ground-states} can be extended to the cases $V \equiv 0$ and $V(x)=|x|^{p}$ with $p>1$. In both cases, we will need a concentration-compactness argument \cite{Li-84} to deal with the lack of compactness. We then find that there exist a $Q\in\mathcal{G}$ and a sequence $\{y_k\}\subset\mathbb R^3$ such that, up to extraction of a subsequence,
	$$
	\lim_{k\to\infty}\Lambda^{-\frac{3}{2}}(a_{*}-a_{k})^{\frac{3}{4}}u_{k}(y_k+\Lambda^{-1}(a_{*}-a_{k})^{\frac{1}{2}}x)=Q(x)
	$$
	strongly in $L^{2}(\mathbb R^3)$, where $\Lambda$ is given by the case $p>1$ of \eqref{lambda-p}.
\end{remark}

\section{\label{sec:energy estimates}Quantum energy estimates}
In this section we settle some energy estimates for the ground state energy and the ground states. Using the ideas of the proof of \eqref{fundamental} in \cite{LiYa-87} we have the following asymptotic formula of the quantum energy

\begin{lemma}\label{lem:many-body-energy}
	Assume that $m>0$ and $V(x)=|x|^p$ for some $p > 0$. Let $a_{N} = a_{*} - N^{-\alpha}$ with $0<\alpha <1/3$. Then we have
	\begin{equation}\label{asymptotic:many-body-energy-2}
	E_{N}^{\rm Q} = E_{a_{N}}^{\rm H} \big(1+o(1)_{N\to\infty}\big) = (a_{*}-a_{N})^{\frac{q}{q+1}} \Big(\frac{q+1}{q}\cdot\frac{\Lambda}{a_{*}} + o(1)_{N\to\infty}\Big)
	\end{equation}
	where $q=\min\{p,1\}$. Furthermore, if $V \equiv 0$ then we also obtain \eqref{asymptotic:many-body-energy-2} with $q=1$ and $\Lambda$ is given by the case $p > 1$ of \eqref{lambda-p}.
\end{lemma}

\begin{proof}
	The upper bound follows from the variational principle
	$$
	E_{N}^{\rm Q} \leq \inf_{u\in H^{\frac{1}{2}}(\mathbb R^3),\|u\|_{L^2}=1}\frac{\left\langle u^{\otimes N},H_{N} u^{\otimes N}\right\rangle}{N} = \inf_{u\in H^{\frac{1}{2}}(\mathbb R^3),\|u\|_{L^2}=1}\mathcal{E}_{a_{N}}^{\rm H}(u) = E_{a_{N}}^{\rm H}.
	$$
	The lower bound was proved by Lieb and Yau (see \cite[Proof of Theorem 2]{LiYa-87}). Their idea is to replace the two-particle interaction by a one-particle interaction. We refer to \cite{Le-15} where the trick is explained in more detailed (see also \cite{LeLe-69} for related arguments). More precisely, it follows from \cite[eq. (2.33)]{LiYa-87} that
	\begin{equation}\label{ineq:lower-bound-1}
	E_{N}^{\rm Q} \geq E_{a_{N}'}^{\rm H} - CN^{-\frac{1}{3}}
	\end{equation}
	where
	\begin{equation}\label{ineq:lower-bound-2}
	a_{N}' = \frac{a_{N}}{(1-CN^{-\frac{1}{3}})(1-N^{-\frac{1}{3}})} < a_{*}.
	\end{equation}
	We deduce from the asymtotic formula of $E_{a_{N}}^{\rm H}$ in \eqref{asymptotic:hartree-energy} that 
	\begin{align}
	E_{a_{N}'}^{\rm H} - E_{a_{N}}^{\rm H} & = \big((a_{*}-a_{N}')^{\frac{q}{q+1}} - (a_{*}-a_{N})^{\frac{q}{q+1}}\big) \Big(\frac{q+1}{q}\cdot\frac{\Lambda}{a_{*}} + o(1)_{N\to\infty}\Big) \nonumber \\
	& \geq -(a_{N}'-a_{N})^{\frac{q}{q+1}} \Big(\frac{q+1}{q}\cdot\frac{\Lambda}{a_{*}} + o(1)_{N\to\infty}\Big). \label{ineq:lower-bound-3}
	\end{align}
	On the other hand, it follows from the formula of $a_{N}'$ in \eqref{ineq:lower-bound-2} that, for large $N$,
	$$
	0< a_{N}' - a_{N} \leq CN^{-\frac{1}{3}}.
	$$
	Thus, we deduce from \eqref{ineq:lower-bound-1} and \eqref{ineq:lower-bound-3} that
	$$
	E_{N}^{\rm Q} \geq E_{a_{N}}^{\rm H} - C N^{-\frac{1}{3}\cdot\frac{q}{q+1}}\Big(\frac{q+1}{q}\cdot\frac{\Lambda}{a_{*}} + o(1)_{N\to\infty}\Big)  = E_{a_{N}}^{\rm H} \big(1 - C N^{-\frac{1}{3}\cdot\frac{q}{q+1}}(a_{*}-a_{N})^{-\frac{q}{q+1}}\big).
	$$
	The error term $N^{-\frac{1}{3}\cdot\frac{q}{q+1}}(a_{*}-a_{N})^{-\frac{q}{q+1}}$ is of order $1$ when $a_{*} - a_{N} = N^{-\alpha}$ for $0<\alpha < 1/3$.
	
	Finally, if $V \equiv 0$ then the asymptotic formula of the quantum energy follows from that of the Hartree energy, which was given in Theorem \ref{thm:behavior-hartree-ground-states}.
\end{proof}

	Lemma \ref{lem:many-body-energy} gave us the convergence of the ground state energy when $N\to\infty$ and $a_{N}$ approaches $a_{*}$ slowly. Note that the case $V \equiv 0$ is allowed in \eqref{asymptotic:many-body-energy-2} as well as in Theorem \ref{thm:behavior-hartree-ground-states} and \ref{thm:behavior-approximate-hartree-ground-states} for the convergence of the (approximate) Hartree ground states. However, in this case and also in the case $V(x) = |x|^{p}$ with $p>1$, the convergence of the (approximate) many-body ground states is a more complicated problem since we need more compactness. In fact, it might fail due to a superposition of them in the linear theory. 
	
	By using \eqref{asymptotic:many-body-energy-2} we now establish some moments estimates for the ground states, which will be useful to control the error made in the mean-field limit. Let us now introduce the following shorthand notation $$h := \sqrt{-\Delta+m^2}+V.$$ Note that $h \geq m>0$. We will need the following technical result, whose proof is given in Appendix \ref{app:operator-bound}.

\begin{lemma}[Operator bounds for two-particle interactions]\label{operator bound}
	We have
	\begin{align}
	|x-y|^{-1} &\leq C(-\Delta_{x})^{\frac{1}{4}}(-\Delta_{y})^{\frac{1}{4}} \label{bound-interaction-1}, \\
	\pm (h_x|x-y|^{-1} + |x-y|^{-1}h_x) & \leq C h_x h_y. \label{bound-interaction-2}
	\end{align}
\end{lemma}

\begin{lemma}[Moments estimates]
	Let $V(x)=|x|^p$ for some $p > 0$ and let $a_{N} = a_{*} - N^{-\alpha}$ with $0<\alpha < 1/3$. Let $\Psi_N\in\mathfrak{H}^N$ be a ground state of $H_{N}$, which exists. Then we have
	\begin{equation}\label{estimate:moments}
	\tr \big(h\gamma_{\Psi_N}^{(1)}\big) \leq C(a_{*}-a_{N})^{-\frac{1}{q+1}},\quad \tr \big(h \otimes h \gamma_{\Psi_N}^{(2)}\big) \leq C(a_{*}-a_{N})^{-\frac{2}{q+1}}
	\end{equation}
	where $q=\min\{p,1\}$. Furthermore, if $0<p\leq 1$ then we have
	\begin{equation}\label{V-estimates}
	\tr \big(V\gamma_{\Psi_N}^{(1)}\big) \leq C(a_{*}-a_{N})^{\frac{p}{p+1}}.
	\end{equation}
\end{lemma}

\begin{proof}
	In what follows, we denote by $h_i$ the operator $h_{x_i}$. For any $0<\varepsilon<1$, we write
	$$
	H_{N} = \varepsilon \sum_{i=1}^N h_{i} + (1-\varepsilon)H_{N,\varepsilon}
	$$
	where the modified Hamiltonian $H_{N,\varepsilon}$ is defined by
	$$
	H_{N,\varepsilon} = \sum_{i=1}^{N} h_{i}-\frac{1}{N-1}\cdot\frac{a_{N}}{1-\varepsilon}\sum_{1\leq i<j \leq N}|x_i-x_j|^{-1}
	$$
	with the corresponding ground state energy $E_{N,\varepsilon}^{\rm Q}$. Since $a_{N} < a_{*}$, we can choose $\varepsilon$ arbitrary small such that $a_{N}/(1-\varepsilon) < a_{*}$. Then it follows from Lemma \ref{lem:many-body-energy} that $E_{N,\varepsilon}^{\rm Q} \geq 0$ for a fixed $N$ large and $a_{N} = a_{*} - N^{-\alpha}$ with $0<\alpha <1/3$. Hence
	$$
	H_{N} \geq \varepsilon \sum_{i=1}^N h_{i}.
	$$
	From this, we obtain that $E_{N}^{\rm Q}>-\infty$. Then the existence of the ground states of $H_{N}$ follows easily from the standard direct method in the calculus of variations.
	
	Assuming that $\Psi_N$ is a ground state of $H_{N}$, we establish its moments estimates. To obtain the first estimate in \eqref{estimate:moments} we choose, in particular, $0<2\varepsilon = (a_{*}-a_{N})/a_{*} < 1$. By taking the expectation against $\Psi_N$ in the above estimate and using the asymptotic formula of $E_{N}^{\rm Q}$ and $E_{N,\varepsilon}^{\rm Q}$ in Lemma \ref{lem:many-body-energy}, we find that
	$$
	\tr \big(h\gamma_{\Psi_N}^{(1)}\big) \leq C\frac{E_{N}^{\rm Q} - (1-\varepsilon)E_{N,\varepsilon}^{\rm Q}}{a_{*}-a_{N}} \leq C(a_{*}-a_{N})^{-\frac{1}{q+1}}.
	$$
	To prove the second estimate in \eqref{estimate:moments}, we process as follow. By the ground state equation
	$$
	H_{N} \Psi_N = NE_{N}^{\rm Q} \Psi_N
	$$
	we can write
	$$
	\frac{1}{2N^2} \Big\langle \Psi_N,\Big(\Big(\sum_{j=1}^{N}h_{j}\Big)H_{N} + H_{N}\Big(\sum_{j=1}^{N}h_{j}\Big)\Big)\Psi_N \Big\rangle = E_{N}^{\rm Q} \tr \big(h\gamma_{\Psi_N}^{(1)}\big).
	$$
	Now, we are after an operator lower bound on
	\begin{align*}
	&\frac{1}{2N^2} \Big(\sum_{j=1}^{N}h_{j}\Big)H_{N} + \frac{1}{2N^2} H_{N}\Big(\sum_{j=1}^{N}h_{j}\Big) \\
	& \quad = \frac{1}{N^2}\Big(\sum_{j=1}^{N}h_{j}\Big)^2 - \frac{a_{N}}{2N^2(N-1)} \sum_{i=1}^{N}\sum_{1\leq j<k \leq N}\big(h_{i} |x_j-x_k|^{-1} + |x_j-x_k|^{-1} h_{i}\big).
	\end{align*}
	For every fixed $i=1,2,\ldots ,N$ we have
	$$
	(1-\varepsilon)H_{N-1,\varepsilon} = (1-\varepsilon)\sum_{j\ne i}^{N} h_{j}-\frac{a_{N}}{N-1}\sum_{i\ne j < k \ne i}|x_j-x_k|^{-1}\geq 0
	$$
	on $\mathfrak{H}^{N-1}$. We can multiply by $h_{i}$ (which commutes with both sides) and then take the sum over $i$. This gives
	$$
	-\frac{a_{N}}{2(N-1)}\sum_{i=1}^{N}\sum_{i \ne j<k \ne i} (h_{i} |x_j - x_k|^{-1} + |x_j - x_k|^{-1} h_{i}) \geq -\Big(\frac{1}{2}+\frac{a_{N}}{2a_{*}}\Big)\sum_{j\ne i}^{N} h_{i} h_{j}.
	$$
	On the other hand, by \eqref{bound-interaction-2} we have
	$$
	-\frac{a_{N}}{2(N-1)} \sum_{j\ne k}(h_{j} |x_j-x_k|^{-1} + |x_j-x_k|^{-1} h_{j}) \geq -\frac{C}{N}\sum_{j\ne k} h_{j} h_{k}.
	$$
	In summary, we found the operator bound
	$$
	\frac{1}{2N^2} \Big(\sum_{j=1}^{N}h_{j}\Big)H_{N} + \frac{1}{2N^2} H_{N}\Big(\sum_{j=1}^{N}h_{j}\Big) \geq \frac{1}{N^2}\Big(\frac{1}{2}-\frac{a_{N}}{2a_{*}}-\frac{C}{N}\Big)\sum_{j\ne k} h_{j} h_{k}.
	$$
	Taking expectation against $\Psi_N$ we obtain
	$$
	E_{N}^{\rm Q} \tr \big(h\gamma_{\Psi_N}^{(1)}\big) \geq \Big(\frac{1}{2}-\frac{a_{N}}{2a_{*}}-\frac{C}{N}\Big) \tr\big(h\otimes h \gamma_{\Psi_N}^{(2)}\big).
	$$
	Thus the second inequality in \eqref{estimate:moments} follows from the first one.
	
	Finally, we prove \eqref{V-estimates} under the assumption $0<p\leq 1$. We first write
	$$
	H_{N} = \frac{1}{2}\sum_{i=1}^N V(x_i) + \tilde{H}_{N}
	$$
	where the modified Hamiltonian $\tilde{H}_{N}$ is defined by
	$$
	\tilde{H}_{N} = \sum_{i=1}^{N}\Big(\sqrt{-\Delta_{x_i}+m^2}+\frac{1}{2}V(x_i)\Big)-\frac{a_{N}}{N-1}\sum_{1\leq i<j \leq N}|x_i-x_j|^{-1}.
	$$
	Since $\Psi_N$ is a ground state of $H_{N}$, we have
	\begin{equation}\label{V-estimates-1}
	E_{N}^{\rm Q} = \frac{\left\langle \Psi_N,H_{N}\Psi_N \right\rangle}{N} \geq \frac{1}{2}\tr\big(V\gamma_{\Psi_N}^{(1)}\big) + \inf\text{spec }\tilde{H}_{N}.
	\end{equation}
	It follows from Theorem \ref{thm:behavior-hartree-ground-states} and Lemma \ref{lem:many-body-energy} that 
	\begin{equation}\label{V-estimates-2}
	\inf\text{spec }\tilde{H}_{N} = (a_{*}-a_{N})^{\frac{p}{p+1}} \Big(\frac{\tilde{\Lambda}}{a_{*}}\cdot\frac{p+1}{p} + o(1)_{N\to\infty}\Big)
	\end{equation}
	where $\tilde{\Lambda}$ is given by the case $0<p\leq 1$ of \eqref{lambda-p}, but with $V(x)$ replaced by $V(x)/2$. It is straightforward that $0<\tilde{\Lambda}<\Lambda$. Therefore, \eqref{V-estimates} is deduced from \eqref{V-estimates-1}, \eqref{V-estimates-2} and the asymptotic formula of $E_{N}^{\rm Q}$ in Lemma \ref{lem:many-body-energy}.
\end{proof}

\begin{remark}
	\begin{itemize}
		\item If $V \equiv 0$ then we also obtain \eqref{estimate:moments} with $q=1$.
		
		\item Having the existence of ground states, we can establish the asymptotic behaviors of the kinetic and interaction energies, along with the asymptotic behavior of the ground state energy \eqref{asymptotic:many-body-energy-2}. More precisely, by applying Lemma \ref{lem:many-body-energy} we can prove that, for $N$ large and $a_{N} = a_{*} - N^{-\alpha}$ with $0<\alpha < 1/3$,
		$$
		\tr\sqrt{-\Delta}\gamma_{\Psi_{N}}^{(1)} \sim (a_{*}-a_{N})^{-\frac{1}{q+1}} ,\quad \iint_{\mathbb R^3 \times \mathbb R^3}\frac{\gamma_{\Psi_{N}}^{(2)}(x,y)}{|x-y|}{\rm d}x{\rm d}y \sim (a_{*}-a_{N})^{-\frac{1}{q+1}}.
		$$
	\end{itemize}
\end{remark}

\section{\label{sec:many body blowup}Many-body blow-up}

	We turn now to the proof of the main result. In our paper, we will use the quantum de Finetti theorem of St{\o}rmer \cite{St-69} and of Hudson and Moody \cite{HuMo-76}. The following formulation is taken from \cite[Corollary 2.4]{LeNaRo-14} (see \cite{Ro-15} for a general discussion and more references).

\begin{theorem}[Quantum de Finetti]\label{Finetti}
	Let $\mathfrak{H}$ be an arbitrary separable Hilbert space and let $\Psi_N\in\bigotimes_{{\rm sym}}^N\mathfrak{H}$ with $\|\Psi_N\|=1$. Assume that the sequence of $k$-particle density matrices $\gamma_{\Psi_N}^{(k)}$ converges strongly in trace class when $N\to\infty$. Then there exists a (unique) Borel probability measure ${\rm d}\mu$ on the unit sphere $S\mathfrak{H}$, invariant under the group action of $\mathcal{S}^1$ such that, up to extraction of a subsequence,
	$$
	\lim_{N\to\infty} \tr\Big|\gamma_{\Psi_N}^{(k)} - \int_{S\mathfrak{H}} |u^{\otimes k}\rangle\langle u^{\otimes k}|{\rm d}\mu(u)\Big| = 0,\quad \forall k\in\mathbb N.$$
\end{theorem}

	We will also use a quantitative version of the quantum de Finetti theorem, originally proved in \cite{ChKoMiRe-07} (see \cite{Ch-11,Ha-13,LeNaRo-15,LeNaRo-16,LeNaRo-17} for variants of the proof and \cite{FaVa-06} for an earlier result in this direction). The following formulation is taken from \cite[Lemma 3.4]{LeNaRo-16}.

\begin{theorem}[Quantitative quantum de Finetti]\label{quantitative-Finetti}
	Let $\Psi\in\mathfrak{H}^N=\bigotimes_{\rm sym}^{N}L^2(\mathbb{R}^3)$ and let $P$ be a finite-rank orthogonal projector with
	$$
	\dim(P\mathfrak{H})=d<\infty.
	$$
	Then there exists a positive Borel measure ${\rm d}\mu_\Psi$ on the unit sphere $SP\mathfrak{H}$ such that
	\begin{equation}\label{ineq:Finetti}
	\tr\Big|P^{\otimes k}\gamma_{\Psi}^{(k)} P^{\otimes k} - \int_{SP\mathfrak{H}}|u^{\otimes k}\rangle\langle u^{\otimes k}|{\rm d}\mu_\Psi(u)\Big| \leq \frac{4kd}{N},\quad \forall k\in\mathbb N.
	\end{equation}
\end{theorem}

We will apply Theorem \ref{quantitative-Finetti} with $P$ a spectral projector below an energy cut-off $L$ for the one-particle operator 
\begin{equation}\label{projector}
P:=\mathbbm{1}(h\leq L)\text{ with }h:=\sqrt{-\Delta + m^2}+V.
\end{equation}
Since $V(x)=|x|^p\in L_{\rm loc}^1(\mathbb R^3)$, the dimension of the low-lying subspace 
$$
d=N_L=\dim(P\mathfrak{H}) = \text{number of eigenvalues of $h$ below $L$}
$$
is finite. Moreover it is controlled by a semi-classical inequality ``\`{a} la Cwikel--Lieb--Rosenblum" stated in the next lemma. This work is due to Daubechies \cite{Da-83} (see also \cite{FrLa-08} and \cite[Theorem 4.2]{LiSe-10} for a thorough discussion of related inequalities).

\begin{lemma}[Low-lying bound states of the one-particle Hamiltonian]\label{lem:dimension} Let $V(x)=|x|^p$ for some $p>0$. Then for $L$ large enough we have
	\begin{equation}\label{ineq:dimension}
	N_L \leq CL^{3+\frac{3}{p}}.
	\end{equation}
\end{lemma}

\begin{proof}
	The number of eigenvalues of $\sqrt{-\Delta+m^2}+V$ below $L$ is smaller than the number of non-positive eigenvalues of $\sqrt{-\Delta}+V-L$, and it can be estimated by
	\begin{align*}
	N_L \leq C\int_{\mathbb R^3}[V(x)-L]_{-}^{3}{\rm d}x = C\int_{|x|\leq L^{\frac{1}{p}}}(L-|x|^p)^{3}{\rm d}x = CL^{3+\frac{3}{p}}.
	\end{align*}
\end{proof}

From Theorem \ref{quantitative-Finetti} and Lemma \ref{lem:dimension} we have the following lower bound for the quantum energy in terms of the quantum de Finetti measure and the second moment.
\begin{lemma}\label{lem:lower-bound-finetti}
	Let $V(x)=|x|^p$ for some $p>0$ and let $\Psi_N\in\mathfrak{H}^N$ be a ground state of $H_{N}$. Let ${\rm d}\mu_{\Psi_N}$ be the de Finetti measure defined in Theorem \ref{quantitative-Finetti} with the projector $P$ as in \eqref{projector}. Then for all $L\geq 1$ we have
	\begin{align}
	E_{N}^{\rm Q} = \frac{\left\langle \Psi_N,H_{N}\Psi_N \right\rangle}{N} & \geq \int_{SP\mathfrak{H}} \mathcal{E}_{a_{N}}^{\rm H}(u){\rm d}\mu_{\Psi_N}(u) - CL\frac{N_L}{N} \nonumber \\
	& \quad - \frac{C}{L^{\frac{1}{4}}}\big(\tr \big(h\gamma_{\Psi_N}^{(1)}\big)\big)^{\frac{1}{4}} \big(\tr \big(h \otimes h \gamma_{\Psi_N}^{(2)}\big)\big)^{\frac{1}{2}}. \label{lower-bound-finetti}
	\end{align}
	Moreover,
	\begin{equation}\label{lower-bound-finetti-1}
	1\geq \int_{SP\mathfrak{H}}{\rm d}\mu_{\Psi_N}(u) \geq \big(\tr\big(P\gamma_{{\Psi_N}}^{(1)}\big)\big)^2 \geq 1 - \frac{2}{L} \tr\big(h\gamma_{\Psi_N}^{(1)}\big).
	\end{equation}
\end{lemma}
Lemma \ref{lem:lower-bound-finetti} is the 3D analogue of \cite[Lemma 4]{LeNaRo-17}. The proof is similar and we omit it for brevity. Now we come to the proof of our main result.

	\begin{proof}[Proof of Theorem \ref{thm:behavior-Nbody}.]
	We assume that $0<p\leq 1$. Inserting the moments estimates \eqref{estimate:moments} into \eqref{lower-bound-finetti} and \eqref{lower-bound-finetti-1}, we obtain
	$$
	E_{N}^{\rm Q} \geq \int_{SP\mathfrak{H}} \mathcal{E}_{a_{N}}^{\rm H}(u){\rm d}\mu_{\Psi_N}(u) - C\frac{L^{4+\frac{3}{p}}}{N} - C\frac{N^{\frac{5\alpha}{4(p+1)}}}{L^{\frac{1}{4}}}
	$$
	and
	$$
	1\geq \int_{SP\mathfrak{H}}{\rm d}\mu_{\Psi_N}(u) \geq \big(\tr\big(P\gamma_{{\Psi_N}}^{(1)}\big)\big)^2 \geq 1 - C\frac{N^{\frac{\alpha}{p+1}}}{L}.
	$$
	It is straightforward to see that if we have addtitionally
	$$
	\alpha<\dfrac{p}{17p+15}
	$$
	then it follows from Lemma \ref{lem:many-body-energy}, with $a_{*}-a_{N}=N^{-\alpha}$, that
	$$E_{N}^{\rm Q} = N^{-\frac{\alpha p}{p+1}} \Big(\frac{p+1}{p}\cdot\frac{\Lambda}{a_{*}} + o(1)_{N\to\infty}\Big).
	$$
	Hence we can choose $L>0$ appropriately such that
	\begin{equation}\label{add-cond}
	\lim_{N\to\infty}\int_{SP\mathfrak{H}} \frac{\mathcal{E}_{a_{N}}^{\rm H}(u)}{E_{a_{N}}^{\rm H}}{\rm d}\mu_{\Psi_N}(u) = \lim_{N\to\infty}\int_{SP\mathfrak{H}}{\rm d}\mu_{\Psi_N}(u) = \lim_{N\to\infty}\tr\big(P\gamma_{{\Psi_N}}^{(1)}\big) = 1.
	\end{equation}
	Since $\mu_{\Psi_N}(SP\mathfrak{H}) = \tr\big(P^{\otimes 2}\gamma_{\Psi_N}^{(2)}P^{\otimes 2}\big)$, we deduce from \eqref{add-cond} that
	\begin{equation}\label{conv-mu}
	1-\mu_{\Psi_N}(SP\mathfrak{H}) = \tr\big(\big(1-P^{\otimes 2}\big)\gamma_{\Psi_N}^{(2)}\big)\leq 2\big(1-\tr\big(P\gamma_{\Psi_N}^{(1)}\big)\big)\to 0.
	\end{equation}
	Therefore, by \eqref{ineq:Finetti}, the triangle inequality and the Cauchy--Schwarz inequality, we also obtain
	\begin{equation}\label{N-conv-1}
	\lim_{N\to\infty}\tr\Big|\gamma_{{\Psi_N}}^{(2)} - \int_{SP\mathfrak{H}} |u^{\otimes 2}\rangle\langle u^{\otimes 2}|{\rm d}\mu_{\Psi_N}(u) \Big| = 0.
	\end{equation}
	Setting $\Phi_N=\ell_{N}^{-3N/2}\Psi_N(\ell_{N}^{-1}\cdot)$ and 
	$$
	\tilde{P} = \mathbbm{1}(\tilde{h}\leq L) \text{ with } \tilde{h}=\ell_{N}\sqrt{-\Delta + m^2\ell_{N}^{-2}}+\ell_{N}^{-p}V.
	$$ 
	It follows from \eqref{N-conv-1} that
	$$
	\lim_{N\to\infty}\tr\Big|\gamma_{{\Phi_N}}^{(2)} - \int_{S\tilde{P}\mathfrak{H}} |u^{\otimes 2}\rangle\langle u^{\otimes 2}|{\rm d}\mu_{\Phi_N}(u)\Big| = 0,
	$$
	which in turn implies that
	\begin{equation}\label{N-conv-2}
	\lim_{N\to\infty}\tr\Big|\gamma_{{\Phi_N}}^{(k)} - \int_{S\tilde{P}\mathfrak{H}} |u^{\otimes k}\rangle\langle u^{\otimes k}|{\rm d}\mu_{\Phi_N}(u)\Big| = 0,\quad \forall k\in\mathbb N.
	\end{equation}

	Now, we denote
	$$
	\delta_{N} = \int_{S\tilde{P}\mathfrak{H}}\Big(\frac{\mathcal{E}_{a_{N}}^{\rm H}(\ell_{N}^{\frac{3}{2}}u(\ell_{N}\cdot))}{E_{a_{N}}^{\rm H}} - 1 \Big){\rm d}\mu_{\Phi_N}(u) = \int_{SP\mathfrak{H}}\Big(\frac{\mathcal{E}_{a_{N}}^{\rm H}(u)}{E_{a_{N}}^{\rm H}}-1\Big){\rm d}\mu_{\Psi_N}(u),
	$$
	then $\delta_{N} \geq 0$ and $\delta_{N} \to 0$ by \eqref{add-cond}. Let $S_{N}$ be the set of all function $u\in H^{\frac{1}{2}}(\mathbb{R}^3)$ satisfying $\|u\|_{L^2}=1$ and 
	\begin{equation}\label{behavior-approximate}
	\frac{\mathcal{E}_{a_{N}}^{\rm H}(\ell_{N}^{\frac{3}{2}}u(\ell_{N}\cdot))}{E_{a_{N}}^{\rm H}} -1 \leq \sqrt{\delta_{N}}.
	\end{equation}
	Let us prove that
	\begin{equation}\label{limsup}
	\lim_{N\to\infty}\sup_{u\in S_{N}}|\left\langle u,v\right\rangle|^{2k}\leq \sup_{u\in\mathcal{G}}|\left\langle u,v\right\rangle|^{2k},\quad \forall v\in L^2(\mathbb R^3),\,\forall k\in\mathbb N.
	\end{equation}
	Assume by contradiction that \eqref{limsup} fails. Then we can find $u_N\in S_{N}$ such that
	\begin{equation}\label{limsup-1}
	\liminf_{N\to\infty}|\left\langle u_N,v\right\rangle|^{2k}> \sup_{u\in\mathcal{G}}|\left\langle u,v\right\rangle|^{2k},\quad \forall v\in L^2(\mathbb R^3),\, \forall k\in\mathbb N.
	\end{equation}
	Since $u_N\in S_{N}$ and $\delta_{N}\to 0$, we deduce from \eqref{behavior-approximate} that
	$$
	\lim_{N\to\infty}\frac{\mathcal{E}_{a_{N}}^{\rm H}(\ell_{N}^{\frac{3}{2}}u_N(\ell_{N}\cdot))}{E_{a_{N}}^{\rm H}} = 1.
	$$
	But then Theorem \ref{thm:behavior-approximate-hartree-ground-states} implies that there exists a $Q\in\mathcal{G}$ such that
	\begin{equation}\label{limsup-2}
	\lim_{N\to\infty}\|u_N-Q\|_{L^2} = 0.
	\end{equation}
	From \eqref{limsup-1} and \eqref{limsup-2} we get
	$$
	|\left\langle Q,v\right\rangle|^{2k}> \sup_{u\in\mathcal{G}}|\left\langle u,v\right\rangle|^{2k},\quad \forall v\in L^2(\mathbb R^3), \forall k\in\mathbb N.
	$$
	This is a contradiction. Hence \eqref{limsup} holds true.
	
	Moreover, by the choice of $S_{N}$ we have
	$$
	\frac{\mathcal{E}_{a_{N}}^{\rm H}(\ell_{N}^{\frac{3}{2}}u(\ell_{N}\cdot))}{E_{a_{N}}^{\rm H}} -1 \geq \sqrt{\delta_{N}},
	$$
	for all $u\in S_{N}^c$. Therefore,
	$$
	\delta_{N} \geq \int_{S_{N}^c} \Big(\frac{\mathcal{E}_{a_{N}}^{\rm H}(u_N)}{E_{a_{N}}^{\rm H}} - 1 \Big){\rm d}\mu_{\Phi_N}(u) \geq \sqrt{\delta_{N}}{\rm d}\mu_{\Phi_N}(S_{N}^c),
	$$
	which yields that ${\rm d}\mu_{\Phi_N}(S_{N}^c)\leq \sqrt{\delta_{N}}\to 0$ and ${\rm d}\mu_{\Phi_N}(S_{N})\to 1$. Thus we conclude from \eqref{N-conv-2} and \eqref{limsup} that for every $v\in L^2(\mathbb R^3)$ and $k\in\mathbb N$,
	\begin{align}
	& \lim_{N\to\infty}\tr \big(|v^{\otimes k}\rangle\langle v^{\otimes k}|\gamma_{\Phi_N}^{(k)}\big) = \lim_{N\to\infty}\int_{S\tilde{P}\mathfrak{H}}|\left\langle u,v\right\rangle|^{2k}{\rm d}\mu_{\Phi_N}(u) \nonumber \\
	& \quad \leq \|v\|_{L^2}^k\lim_{N\to\infty}{\rm d}\mu_{\Phi_N}(S_{N}^c) + \lim_{N\to\infty}{\rm d}\mu_{\Phi_N}(S_{N})\lim_{N\to\infty}\sup_{u\in S_{N}}|\left\langle u,v\right\rangle|^{2k} \nonumber  \\
	& \quad \leq \sup_{u\in\mathcal{G}}|\left\langle u,v\right\rangle|^{2k}. \label{limsup-3}
	\end{align}
	
	On the other hand, we infer from \eqref{estimate:moments} and \eqref{V-estimates} that
	\begin{equation}\label{estimate-Finetti}
	\tr\big(\big(\sqrt{-\Delta}+V\big)\gamma_{\Phi_N}^{(1)}\big) \leq C.
	\end{equation}
	Since $\sqrt{-\Delta}+V$ has a compact resolvent, \eqref{estimate-Finetti} implies that, up to extraction of a subsequence, $\gamma_{\Phi_N}^{(1)}$ converges to $\gamma^{(1)}$ strongly in the trace class. Modulo a diagonal extraction argument, one can assume that the convergence is along the same subsequence. By \cite[Corollary 2.4]{LeNaRo-14}, $\gamma_{\Phi_N}^{(k)}$ converges to $\gamma^{(k)}$ strongly as well for all $k\geq 1$. By the quantum de Finetti Theorem \ref{Finetti}, there exists a Borel probability measure ${\rm d}\mu$ on the unit sphere $S\mathfrak{H}$ such that
	$$
	\gamma^{(k)}= \int_{S\mathfrak{H}} |u^{\otimes k}\rangle\langle u^{\otimes k}|{\rm d}\mu(u),\quad \forall k\in\mathbb N.
	$$
	To complete the proof, we will show that ${\rm d}\mu$ is supported on $\mathcal{G}$. From \eqref{limsup-3} and the strong convergence $\gamma_{\Phi_N}^{(k)}\to\gamma^{(k)}$ in trace class, we get 
	\begin{equation}\label{limsup-4}
	\int_{S\mathfrak{H}}|\left\langle u,v\right\rangle|^{2k}{\rm d}\mu(u) \leq \sup_{u\in\mathcal{G}}|\left\langle u,v\right\rangle|^{2k},\quad \forall v\in L^2(\mathbb R^3), \forall k\in\mathbb N.
	\end{equation}
	We assume by contradiction that there exists $v_0$ in the support of ${\rm d}\mu$ (the unit sphere $S\mathfrak{H}$) and $v_0 \notin \mathcal{G}$. We claim that we could then find $\eta\in (0,1/2)$ such that
	\begin{equation}\label{limsup-5}
	|\left\langle u,v\right\rangle| \leq 1-3\eta^2,\quad \forall u\in\mathcal{G},\forall v\in D
	\end{equation}
	where
	$$
	D=\{v\in \text{supp }{\rm d}\mu = S\mathfrak{H}: \|v-v_0\|_{L^2}<\eta\}.
	$$
	Indeed, if that were not the case, we would have two sequences strongly converging in $L^2(\mathbb R^3)$
	$$
	u_n\to u_0\in\mathcal{G}\quad \text{and} \quad v_n\to v_0
	$$
	with $\|u_n-v_n\|_{L^2}\to 0$. This implies that $v_0\in\mathcal{G}$, a contradiction. Here we have used the fact that $\mathcal{G}$ is a compact subset of $L^2(\mathbb R^3)$.
	
	On the other hand, by the triangle inequality, we have
	\begin{equation}\label{limsup-6}
	|\left\langle u,v\right\rangle| \geq \frac{\|u\|_{L^2}^2+\|v\|_{L^2}^2-\|u-v\|_{L^2}^2}{2}\geq 1-2\eta^2,\quad \forall u,v\in D.
	\end{equation}
	Combining \eqref{limsup-4}, \eqref{limsup-5} and \eqref{limsup-6} we find that
	\begin{align*}
	(\mu(D))^2(1-2\eta^2)^{2k} & \leq \int_D\int_D |\left\langle u,v\right\rangle|^{2k}{\rm d}\mu(u){\rm d}\mu(v) \\
	& \leq \int_D \sup_{u\in\mathcal{G}}|\left\langle u,v\right\rangle|^{2k}{\rm d}\mu(v) \leq \mu(D)(1-3\eta^2)^{2k},\quad  \forall k\in\mathbb N.
	\end{align*}
	By taking $k\to\infty$ we obtain that ${\rm d}\mu(D)=0$. However, it is a contradiction to the fact that $v_0$ belongs to the support of ${\rm d}\mu$ and ${\rm d}\mu$ is a Borel measure. Thus we conclude that ${\rm d}\mu$ is supported on $\mathcal{G}$ and the proof is completed.
	\end{proof}
	
	\section*{Acknowledgements} The author is indebted to P.T. Nam for helpful discussions. Also, he is very grateful to the referee for many useful suggestions which improved significantly the presentation of the paper. He cordially thanks J. Ricaud for his careful reading of the manuscript.
	
	\appendix
	
	\section{Proof of Lemma \ref{operator bound}\label{app:operator-bound}}
	
	In this appendix we prove Lemma \ref{operator bound}. For the proof of which we will need the following standard lemma (see \cite[Appendix B]{LaLuNa-19} for a detailed proof and more general statements).
	
	\begin{lemma}\label{lem:density}
		Let $\mathbb{D}=\{(x,x):x\in\mathbb R^3\}$ be the diagonal in $\mathbb R^3 \times \mathbb R^3$. Then $C_c^{\infty}((\mathbb R^3 \times \mathbb R^3)\backslash \mathbb{D})$ is dense in $H^1(\mathbb R^3 \times \mathbb R^3)$.
	\end{lemma}
	
	\begin{proof}[Sketch of the proof]
		Note that any function $u\in C_c^\infty((\mathbb R^3 \times \mathbb R^3)\backslash \mathbb{D})$ can be interpreted as a function in $C_c^\infty(\mathbb R^3 \times \mathbb R^3)$ by a natural extension $u(x,x)=0$ for all $x\in \mathbb R^3$. For every $u\in H^1(\mathbb R^3 \times \mathbb R^3)$, since $C_c^\infty(\mathbb R^3 \times \mathbb R^3)$ is dense in $H^1(\mathbb R^3 \times \mathbb R^3)$, there exists a sequence $\phi_n \in C_c^\infty(\mathbb R^3 \times \mathbb R^3)$ such that $\phi_n \to u$ in $H^1(\mathbb R^3 \times \mathbb R^3)$ as $n\to \infty$. Let $g \in C_c^\infty(\mathbb R^3)$ be a fixed function such that $0\leq g\leq 1$, $g(x) = 0$ if $|x|\leq 1$ and $g(x) = 1$ if $|x|\geq 2$. For $(x,y)\in \mathbb R^3\times \mathbb R^3$, we define $$
		u_n(x,y)=g(n(x-y))\phi_n(x,y).
		$$
		Then one can verify that $u_n \in C_c^\infty((\mathbb R^3 \times \mathbb R^3)\backslash \mathbb{D})$ and that $u_n \to u$ strongly in $H^1(\mathbb R^3 \times \mathbb R^3)$.
	\end{proof}
	
	Since $H^{1}(\mathbb R^3 \times \mathbb R^3)\subset H^{\frac{1}{2}}(\mathbb R^3 \times \mathbb R^3)$ densely (see \cite[Proof of Theorem 7.14]{LiLo}), we deduce from Lemma \ref{lem:density} that $C_c^{\infty}((\mathbb R^3 \times \mathbb R^3)\backslash \mathbb{D})$ is dense in $H^{\frac{1}{2}}(\mathbb R^3 \times \mathbb R^3)$. Now we can provide the announced proof.
	
	\begin{proof}[Proof of Lemma \ref{operator bound}.]
	We are going to prove that the following inequalities
	\begin{align}
	|x-y|^{-4} & \leq C(-\Delta_{x})(-\Delta_{y}), \label{bound-interaction-3} \\
	(h_x |x-y|^{-1} + |x-y|^{-1} h_x)^2 & \leq C h_{x}^{2} h_{y}^{2} \label{bound-interaction-4}
	\end{align}
	hold true in $C_c^{\infty}((\mathbb R^3 \times \mathbb R^3)\backslash \mathbb{D})$ where $\mathbb{D}=\{(x,x):x\in\mathbb R^3\}$ is the diagonal in $\mathbb R^3 \times \mathbb R^3$. Then \eqref{bound-interaction-1} and \eqref{bound-interaction-2} are deduced respectively from \eqref{bound-interaction-3} and \eqref{bound-interaction-4}, by Lemma \ref{lem:density} and the fact that the function $t\mapsto t^s$ is operator monotone for any $0\leq s \leq 1$ (see \cite{Lo-34, He-51, Si-11, Do-74}).
	
	We first prove \eqref{bound-interaction-3}. To deal with the term $|x-y|^{-4}$, we use Hardy's inequality (see e.g. \cite{Ya-99,FrLiSe-08})
	\begin{equation}\label{Hardy}
	\frac{1}{4|x|^2} \leq -\Delta_{x}.
	\end{equation}
	For every $f\in C_c^{\infty}((\mathbb R^3 \times \mathbb R^3)\backslash \mathbb{D})$, by applying \eqref{Hardy} in the variable $x$ with $y$ fixed, we have
	\begin{align}
	\left\langle f,|x-y|^{-4}f \right\rangle & \leq 4\left\langle |x-y|^{-1}f,-\Delta_{x}(|x-y|^{-1}f) \right\rangle \label{hardy-type-1} \\
	& = 4\iint_{\mathbb R^3 \times \mathbb R^3} |\nabla_{x} (|x-y|^{-1}f(x,y))|^2{\rm d}x{\rm d}y \nonumber \\
	& = 4\iint_{\mathbb R^3 \times \mathbb R^3} |\nabla_{x} |x-y|^{-1}|^2 |f(x,y)|^2 + |x-y|^{-2}|\nabla_{x} f(x,y)|^2{\rm d}x{\rm d}y \nonumber \\
	& \quad + 8\Re\iint_{\mathbb R^3 \times \mathbb R^3} |x-y|^{-1}f(x,y)\nabla_{x} |x-y|^{-1}\nabla_{x} \overline{f(x,y)}{\rm d}x{\rm d}y. \label{hardy-type}
	\end{align}
	A calculation using integration by part gives us
	\begin{align}
	& \Re\iint_{\mathbb R^3 \times \mathbb R^3} |x-y|^{-1}f(x,y)\nabla_{x} |x-y|^{-1}\nabla_{x} \overline{f(x,y)}{\rm d}x{\rm d}y \nonumber \\
	& \quad = -\Re\iint_{\mathbb R^3 \times \mathbb R^3} \overline{f(x,y)}\nabla_{x}(|x-y|^{-1}f(x,y)\nabla_{x}|x-y|^{-1}) {\rm d}x{\rm d}y \nonumber \\
	& \quad = - \iint_{\mathbb R^3 \times \mathbb R^3} |f(x,y)|^2 (|\nabla_{x} |x-y|^{-1}|^2+|x-y|^{-1}\Delta_{x} |x-y|^{-1}){\rm d}x{\rm d}y \nonumber \\
	& \qquad - \Re\iint_{\mathbb R^3 \times \mathbb R^3} |x-y|^{-1}\nabla_{x} |x-y|^{-1}\overline{f(x,y)}\nabla_{x} f(x,y){\rm d}x{\rm d}y.\label{laplace-0}
	\end{align}
	Note that $\Delta_{x} |x-y|^{-1} = \delta_y(x) = 0$ in $(\mathbb R^3 \times \mathbb R^3)\backslash \mathbb{D}$. We thus deduce from \eqref{laplace-0} that
	\begin{align}
	& 2\Re\iint_{\mathbb R^3 \times \mathbb R^3} |x-y|^{-1}f(x,y)\nabla_{x} |x-y|^{-1}\nabla_{x} \overline{f(x,y)}{\rm d}x{\rm d}y \nonumber \\
	& \quad = -\iint_{\mathbb R^3 \times \mathbb R^3} |\nabla_{x} |x-y|^{-1}|^2 |f(x,y)|^2{\rm d}x{\rm d}y. \label{hardy-type-0}
	\end{align}
	Inserting \eqref{hardy-type-0} into \eqref{hardy-type} and using Hardy's inequality \eqref{Hardy} applied in the variable $y$ with $x$ fixed, we get
	\begin{equation}\label{hardy-type-2}
	\left\langle f,|x-y|^{-4}f \right\rangle \leq 4\iint_{\mathbb R^3 \times \mathbb R^3} |x-y|^{-2}|\nabla_{x} f(x,y)|^2{\rm d}x{\rm d}y
	\leq 16\langle f,(-\Delta_{x})(-\Delta_{y})f \rangle.
	\end{equation}
	Hence \eqref{bound-interaction-3} holds true in $C_c^\infty((\mathbb R^3 \times \mathbb R^3)\backslash \mathbb{D})$. Since the function $t\mapsto \sqrt[4]{t}$ is operator monotone, it follows from \eqref{hardy-type-2} that 
	$$
	\left\langle f,|x-y|^{-1}f \right\rangle \leq 2\langle f,(-\Delta_{x})^{\frac{1}{4}}(-\Delta_{y})^{\frac{1}{4}}f \rangle ,\quad \forall f\in C_c^\infty((\mathbb R^3 \times \mathbb R^3)\backslash \mathbb{D}).
	$$
	For $f\in H^{\frac{1}{2}}(\mathbb R^3 \times \mathbb R^3)$, there exists a sequence $\{f_n\}\subset C_c^\infty((\mathbb R^3 \times \mathbb R^3)\backslash \mathbb{D})$ such that $f_n\to f$ in $H^{\frac{1}{2}}(\mathbb R^3 \times \mathbb R^3)$, by Lemma \ref{lem:density}. Hence, up to extraction of a subsequence, $f_n \to f$ pointwise in $\mathbb R^3 \times \mathbb R^3$. By Fatou's Lemma we have
	\begin{align*}
	\left\langle f,|x-y|^{-1}f \right\rangle & \leq \lim_{n\to\infty} \left\langle f_n,|x-y|^{-1}f_n \right\rangle \\
	& \leq 2 \lim_{n\to\infty} \langle f_n,(-\Delta_{x})^{\frac{1}{4}}(-\Delta_{y})^{\frac{1}{4}}f_n \rangle = 2\langle f,(-\Delta_{x})^{\frac{1}{4}}(-\Delta_{y})^{\frac{1}{4}}f \rangle.
	\end{align*}
	This completes the proof of \eqref{bound-interaction-3}.
	
	Now we come to prove \eqref{bound-interaction-4}. Observing that, by the inequality for operators
	\begin{equation}\label{bound-interaction-2-1}
	(A+A^{*})^2 \leq 2(AA^{*}+A^{*}A),
	\end{equation}
	it is enough to prove the following inequality
	\begin{equation}\label{bound-interaction-2-2}
	h_x |x-y|^{-2} h_x + |x-y|^{-1} h_{x}^{2} |x-y|^{-1} \leq C h_{x}^{2} h_{y}^{2}.
	\end{equation}
	Applying \eqref{Hardy} in the variable $y$ with $x$ fixed, and note that $h_x$ and $h_y$ commute, we obtain
	\begin{equation}\label{bound-interaction-2-3}
	h_x |x-y|^{-2} h_x \leq 4 h_x (-\Delta_{y}) h_x \leq 4 h_x h_{y}^{2} h_x = 4 h_{x}^{2} h_{y}^{2}.
	\end{equation}
	On the other hand, by Cauchy--Schwarz inequality we have
	\begin{equation}\label{bound-interaction-2-4}
	|x-y|^{-1} h_{x}^{2} |x-y|^{-1} \leq 2 |x-y|^{-1} (-\Delta_{x}+m^2+V(x)^2) |x-y|^{-1}.
	\end{equation}
	Again, applying \eqref{Hardy} in the variable $y$ with $x$ fixed, we obtain
	\begin{equation}\label{bound-interaction-2-5}
	|x-y|^{-1} (m^2+V(x)^2) |x-y|^{-1} \leq 4 h_{x}^{2} (-\Delta_{y}) \leq 4 h_{x}^{2} h_{y}^{2}.
	\end{equation}
	Furthermore, it follows easily from \eqref{hardy-type-1}--\eqref{hardy-type-2} that
	\begin{equation}\label{bound-interaction-2-6}
	|x-y|^{-1} (-\Delta_{x}) |x-y|^{-1} \leq 4 (-\Delta_x)(-\Delta_{y}) \leq 4 h_{x}^{2} h_{y}^{2}.
	\end{equation}
	Combining \eqref{bound-interaction-2-1}--\eqref{bound-interaction-2-6}, we obtain that \eqref{bound-interaction-4} holds true in $C_c^\infty((\mathbb R^3 \times \mathbb R^3)\backslash \mathbb D)$. This implies that \eqref{bound-interaction-2} holds true in $C_c^\infty((\mathbb R^3 \times \mathbb R^3)\backslash \mathbb D)$, by the operator monotone function $t\mapsto \sqrt{t}$. Since $C_c^\infty((\mathbb R^3 \times \mathbb R^3)\backslash \mathbb D)$ is dense in $H^1(\mathbb R^3 \times \mathbb R^3)$, by Lemma \ref{lem:density}, we deduce from this that \eqref{bound-interaction-2} holds true in $H^1(\mathbb R^3 \times \mathbb R^3)$.
	\end{proof}
	
	\begin{remark}
		It follows from our proof of Lemma \ref{operator bound} that
		$$
		|x-y|^{-4s} \leq 2^{4s}(-\Delta_{x})^{s}(-\Delta_{y})^{s},\quad \forall 0\leq s\leq 1
		$$
		in $H^{2s}(\mathbb R^3 \times \mathbb R^3)$. When $s=1$, the above inequality is sharp and the \emph{ground state} is $|x-y|^{-\frac{1}{2}}$. This could be compared to Hardy's inequality.
	\end{remark}

\end{document}